\documentclass{article}

\usepackage{a4}

\scrollmode
\usepackage{amsmath}
\usepackage{amsfonts}
\usepackage{amssymb}
\usepackage{latexsym}
\usepackage{stmaryrd}
\usepackage{array}
\usepackage{exscale}
\newcommand{\nc}{\newcommand}
\newcommand{\ol}{\overline}

\newcommand{\es}{\emptyset}
\newcommand{\sm}{\setminus}
\newcommand{\ve}{\varepsilon}
\newcommand{\vp}{\varphi}

\newcommand{\bc}{\bigcup}

\newcommand{\ra}{\rightarrow}

\newcommand{\sse}{\subseteq}

\newcommand{\mr}{\mathrm}
\newcommand{\mc}{\mathcal}
\newcommand{\mf}{\mathfrak}

\newcommand{\DMO}{\DeclareMathOperator}
\newcommand{\DST}{\displaystyle}

\newcommand{\ZZ}{\mathbb{Z}}
\newcommand{\NN}{\mathbb{N}}
\newcommand{\NNZ}{\NN_0}

\newcommand{\RR}{\mathbb{R}}

\newcommand{\PP}{\mathbb{P}}

\newcommand{\Ende}{\ \rule{0.4em}{1.7ex}}

\usepackage{listings}
\newcommand{\inl}[1]{\lstinline$#1$}
\newcommand{\set}[1]{\{ #1 \}}
\newcommand{\setb}[1]{\big \{ \, #1 \, \big \}}
\DMO{\dom}{dom}
\DMO{\id}{id}
\DMO{\cod}{cod} 
\DMO{\rg}{rg} 
\DMO{\tcomp}{\trans{\circ}} 
\DMO{\simrv}{\,\sim\hspace{-0.05em}}
\DMO{\simlv}{\!\sim\,}
\nc{\simlvi}[1]{\!\sim_{#1}}
\DMO{\card}{card}
\DMO{\proj}{pr}
\DMO{\inj}{in}
\DMO{\symdif}{\vartriangle} 
\DMO{\addcup}{{\stackrel{\text{\raisebox{-2.2ex}[-0ex][-0ex]{\large$\cdot$}}}{\cup}}} 
\DMO{\addbcup}{{\stackrel{\text{\raisebox{-4.2ex}[-0ex][-0ex]{\Large$\cdot$}}}{\bigcup}}} 
\nc{\apprel}[3]{{#1}(#2)_{(#3)}} 
\DMO{\Rel}{\mf{REL}} 
\DMO{\Abb}{\mf{MAP}} 
\DMO{\Tra}{\mf{T}} 
\DMO{\Per}{\mf{S}} 
\DMO{\Pert}{\Per_t} 
\DMO{\Ptr}{\mf{PT}} 
\DMO{\fix}{fix} 
\DMO{\Peri}{\Per_i} 
\DMO{\Pers}{\Per_s} 
\DMO{\Rrel}{\Rel_r} 
\DMO{\Srel}{\Rel_s} 
\DMO{\Trel}{\Rel_t} 
\DMO{\konkat}{\sqcup} 
\DMO{\cmpl}{\complement^1} 
\nc{\cmpli}[1]{\complement^1_{#1}} 
\DMO{\cmplz}{\complement^0} 
\nc{\cmplzi}[1]{\complement^0_{#1}} 
\DMO{\cmplzo}{\complement^*} 
\nc{\cmplzoi}[1]{\complement^*_{#1}} 
\DMO{\fsigma}{{\mf{F}}_{\sigma}} 
\DMO{\gdeltao}{\mf{G}_{\sigma}}
\DMO{\fs}{{\mf{F}}_{s}} 
\DMO{\fss}{{\mf{F}}_{s}^*} 
\nc{\zf}{\mr{ZF}}
\nc{\zfmf}{\zf^0} 
\nc{\zfc}{\mr{ZFC}}
\nc{\zfcmf}{\zfc^0} 
\nc{\bst}{\mr{BST}} 
\newcommand{\tb}[2]{\set{#1, \dots, #2}} 
\DMO{\re}{Re}
\DMO{\im}{Im}
\DMO{\sgn}{sgn} 
\providecommand{\abs}[1]{\lvert #1 \rvert} 
\DMO{\ld}{ld} 
\nc{\untit}[2]{{#1}^{#2 \downarrow}} 
\nc{\obit}[2]{{#1}^{#2 \uparrow}} 
\DMO{\inttop}{\tau_{\mr{O}}} 
\DMO{\rointtop}{\tau_+} 
\DMO{\lointtop}{\tau_-} 
\DMO{\sid}{\mc{IDL}} 
\DMO{\skid}{\mc{CID}} 
\DMO{\smid}{\sid_m} 
\DMO{\smkid}{\skid_m} 
\DMO{\nachbarn}{\Gamma}
\DMO{\enachbarn}{N}
\DMO{\nachbarnr}{\Gamma_{\!\mr{r}}}
\DMO{\nachbarnz}{\widetilde{\Gamma}}
\DMO{\nachbarnzr}{\widetilde{\Gamma}_{\!\mr{r}}}
\DMO{\inzEK}{\mc{I}^{\mr{V}}}
\DMO{\inzEKe}{\mc{I}^{\mr{V}}_1}
\DMO{\inzEKz}{\mc{I}^{\mr{V}}_2}
\nc{\inzEKi}[1]{\mc{I}^{\mr{V}}_{#1}}
\DMO{\inzKE}{\mc{I}^{\mr{E}}}
\DMO{\inzKEe}{\mc{I}^{\mr{E}}_1}
\DMO{\inzKEz}{\mc{I}^{\mr{E}}_2}
\nc{\inzKEi}[1]{\mc{I}^{\mr{E}}_{#1}}
\DMO{\inz}{I}
\DMO{\tinz}{\trans{\inz}} 
\DMO{\adjE}{\mc{A}^{\mr{V}}}
\DMO{\adjEe}{\mc{A}^{\mr{V}}_1}
\DMO{\adjEz}{\mc{A}^{\mr{V}}_2}
\nc{\adjEi}[1]{\mc{A}^{\mr{V}}_{#1}}
\DMO{\adjor}{\mc{A}_{\mr{S}}} 
\DMO{\adjK}{\mc{A}^{\mr{E}}}
\DMO{\adj}{A}
\DMO{\degmin}{\overleftarrow{\deg}}
\DMO{\degmax}{\overrightarrow{\deg}}
\DMO{\degdur}{\widetilde{\deg}} 
\DMO{\ideg}{idg} 
\DMO{\odeg}{odg} 
\DMO{\degmaxl}{\overrightarrow{\deg}_{<}}
\DMO{\degl}{\deg_{<}}
\DMO{\rankmin}{\overleftarrow{\rank}}
\DMO{\rankmax}{\overrightarrow{\rank}}
\DMO{\rankdur}{\widetilde{\rank}}
\DMO{\rankmaxl}{\overrightarrow{\rank}_{<}}
\DMO{\rankl}{\rank_{<}}
\DMO{\vertexcon}{\kappa} 
\DMO{\edgecon}{\lambda} 
\DMO{\treewidth}{tw} 
\DMO{\girth}{g} 
\DMO{\circumference}{cf} 
\DMO{\length}{lgth} 
\DMO{\npm}{\Phi} 
\DMO{\concomp}{cc} 
\DMO{\nconcomp}{ncc} 
\DMO{\indprim}{ip} 
\DMO{\indimprim}{iip} 
\DMO{\bouquet}{B}
\DMO{\dipol}{D}
\DMO{\jkg}{J} 
\DMO{\vjkg}{VK} 
\DMO{\Tr}{Tr} 
\DMO{\Ind}{Ind} 
\DMO{\Zuo}{Mat} 
\DMO{\Pzuo}{PMat} 
\DMO{\St}{St} 
\DMO{\Ints}{Ints} 
\DMO{\Cov}{Cov} 
\DMO{\closse}{clo_{\sse}} 
\DMO{\clospe}{clo_{\supseteq}} 
\DMO{\edgemg}{ML} 
\DMO{\kneserg}{K} 
\DMO{\knesern}{\tau_0} 
\DMO{\nis}{nis} 
\DMO{\PBD}{PBD}
\nc{\BD}[1]{{#1}\text{-}\mr{BD}}
\DMO{\BIBD}{BIBD}
\DMO{\Steiner}{S}
\DMO{\SteinerTriple}{STS}
\DMO{\SteinerQuadruple}{SQS}
\DMO{\progeo}{PG} 
\DMO{\affgeo}{AG} 

\DMO{\astriv}{A_t} 
\DMO{\KochenSpecker}{KS}
\DMO{\KochenSpeckerErw}{KS'}
\DMO{\rankd}{rd}
\DMO{\mnconcomp}{mncc}
\DMO{\gpk}{\Box} 
\DMO{\gpw}{\times} 
\DMO{\gps}{\boxtimes} 
\DMO{\gjoin}{\boxdot} 
\DMO{\gjoinplus}{\boxplus} 
\DMO{\Ketten}{\mc{L}}
\DMO{\Antiketten}{\mc{A}}
\DMO{\comparable}{\Bumpeq}
\DMO{\incomparable}{\parallel}
\DMO{\pot}{\PP} 
\DMO{\pote}{\PP_f} 
\DMO{\potfv}{\overrightarrow{\PP}} 
\DMO{\potfvn}{\overrightarrow{\PP}^{\!*}} 
\DMO{\potfr}{\overleftarrow{\PP}} 
\DMO{\fak}{fac}
\newcommand{\floor}[1]{\lfloor #1 \rfloor}

\DMO{\partitiont}{p}
\DMO{\teilt}{\mid} 
\DMO{\nteilt}{\nmid} 
\nc{\Prim}{\mc{PR}} 
\DMO{\ord}{ord}
\DMO{\ggt}{ggt}
\DMO{\kgv}{kgv}
\DMO{\opmod}{mod}
\DMO{\opdiv}{div}
\DMO{\eulphi}{\vp}
\DMO{\Li}{Li}
\DMO{\Ei}{Ei}
\newcommand{\Cls}{\mc{CLS}}

\newcommand{\Sat}{\mc{SAT}}

\newcommand{\Musat}{\mc{M\hspace{0.8pt}U}} 
\newcommand{\Musati}[1]{\Musat_{\!#1}} 
\newcommand{\Smusat}{\mc{S}\Musat} 
\newcommand{\Smusati}[1]{\Smusat_{\!#1}}

\nc{\Clsoo}{\Cls^{1,1}} 
\DMO{\munpuclash}{\mu{}NH}
\DMO{\hdef}{\delta_{\mr{h}}} 
\DMO{\rdef}{\delta_{\mr{r}}} 

\newcommand{\Lean}{\mc{LEAN}}

\newcommand{\Mlean}{\mc{M}\Lean}
\newcommand{\Mleani}[1]{\Mlean_{\!#1}}

\DMO{\nulli}{null} 
\DMO{\lit}{lit}
\DMO{\var}{var}
\DMO{\val}{val}
\DMO{\res}{\diamond} 
\DMO{\resop}{Res} 
\DMO{\mresop}{mRes} 
\DMO{\dpl}{DP} 

\DMO{\comp}{Comp} 
\DMO{\compex}{\comp_{ER}} 
\DMO{\compr}{\comp_R} 
\DMO{\comptr}{\comp_{tR}} 
\newcommand{\Us}{\mc{U}} 
\DMO{\comptru}{\comp_{tR(\Us)}} 
\DMO{\compru}{\comp_{R(\Us)}}

\DMO{\hardness}{hd}

\DMO{\pebf}{PF} 
\DMO{\rt}{rt} 
\DMO{\nds}{nds} 
\DMO{\lvs}{lvs} 
\DMO{\nlvs}{\#lvs} 
\DMO{\nnds}{\#nds} 
\DMO{\height}{ht}
\DMO{\depth}{d}
\DMO{\cls}{cls}
\DMO{\newcommandls}{\#cls}
\DMO{\ds}{ds}
\DMO{\dst}{ds_T}
\DMO{\dsg}{ds_G}
\DMO{\dpr}{dp}
\DMO{\dprt}{dp_T}
\DMO{\dprg}{dp_G}
\DMO{\ind}{in}
\DMO{\indg}{in_G}
\DMO{\outd}{out}
\DMO{\outdg}{out_G}
\DMO{\peb}{peb} 

\newcommand{\pao}[2]{\langle #1 \ra #2 \rangle}

\DMO{\taum}{\max \tau}
\DMO{\tauprob}{\tau^p} 
\DMO{\mtau}{\mf{T}} 
\DMO{\concatbt}{;} 
\DMO{\compobt}{\merge} 
\nc{\bth}[1]{\langle{#1}\rangle} 
\DMO{\pc}{pc}
\DMO{\aut}{Auk} 
\DMO{\laut}{LAuk} 
\DMO{\lautz}{LAuk_0} 
\DMO{\maut}{MAuk}
\newcommand{\A}{\mc{A}} 
\DMO{\nv}{N} 
\DMO{\na}{\nv_a} 
\DMO{\nA}{\nv_{\A}} 
\DMO{\nla}{\nv_{la}}
\DMO{\nbla}{\nv_{bla}}
\DMO{\nma}{\nv_{ma}}
\DMO{\npa}{\nv_{pa}}
\DMO{\baut}{BAuk} 
\DMO{\blaut}{BLAuk} 
\DMO{\blautz}{BLAuk_0} 
\DMO{\paut}{PAut} 
\DMO{\pautz}{PAut_0} 

\DMO{\resouz}{\overset{\Us, 0}{\vdash}}
\DMO{\resouo}{\overset{\Us, 1}{\vdash}}
\DMO{\resouk}{\overset{\Us,\, k}{\vdash}}
\DMO{\resou}{\,\overset{\Us}{\vdash}\,}
\DMO{\resour}{\,\overset{\Us_0}{\vdash}\,}
\DMO{\resourz}{\,\overset{\Us_0, 0}{\vdash}\,}
\DMO{\uresouk}{\resouk\hspace{-0.6em}\mbox{\raisebox{0.8ex}{\tiny u}}}
\DMO{\bresouk}{\resouk\hspace{-0.6em}\mbox{\raisebox{0.8ex}{\tiny b}}}
\DMO{\iresouk}{\resouk\hspace{-0.6em}\mbox{\raisebox{0.8ex}{\tiny i}}}
\DMO{\resok}{\overset{k}{\vdash}} 

\DMO{\wid}{wid} 
\DMO{\widl}{\hspace*{-1.5pt}wid}
\DMO{\widb}{\sideset{^{\mr{b}}}{}\widl}
\DMO{\widi}{\sideset{^{\mr{i}}}{}\widl}
\DMO{\cwid}{\mc{W}} 
\DMO{\cwidl}{\hspace*{-1pt}\mc{W}} 
\DMO{\cwidb}{\sideset{^{\mr{b}}}{}\cwidl}
\DMO{\cwidi}{\sideset{^{\mr{i}}}{}\cwidl}
\DMO{\modp}{mod_p} 
\DMO{\modt}{mod_t} 
\DMO{\moda}{\mf{S}} 
\DMO{\modf}{fal} 
\DMO{\mods}{mod} 
\DMO{\mus}{MU}
\DMO{\mss}{MS}
\DMO{\cmus}{CMU}
\DMO{\cmss}{CMS}
\DMO{\eqs}{EQ} 
\DMO{\neqs}{NEQ} 
\DMO{\scf}{CM} 
\DMO{\acf}{DCM} 
\DMO{\cmg}{cmg} 
\DMO{\cmdg}{cmdg} 
\DMO{\cg}{cg} 
\DMO{\gcg}{cgg} 
\DMO{\gcdg}{cgdg} 
\DMO{\rsg}{rg} 
\DMO{\srsg}{srg} 
\DMO{\vhg}{vhg} 
\DMO{\cvg}{cvg} 
\DMO{\cvmg}{cvmg} 
\DMO{\vig}{vig} 
\DMO{\vcg}{vcg} 
\DMO{\nscf}{bcp} 
\DMO{\nacf}{bcp_d} 
\DMO{\bcp}{bcp} 
\DMO{\tbcp}{tbcp} 
\DMO{\nsat}{\#sat}
\DMO{\nusat}{\#usat}
\DMO{\maxsat}{maxsat}
\DMO{\pmin}{\rankmin}
\DMO{\pmax}{\rankmax}
\DMO{\pav}{\rankdur}
\DMO{\ldeg}{ld} 
\DMO{\minldeg}{\mu\!\ldeg} 
\DMO{\maxldeg}{\nu\ldeg} 
\DMO{\vdeg}{vd} 
\DMO{\minvdeg}{\mu\!\vdeg} 
\DMO{\maxvdeg}{\nu\!\vdeg} 
\DMO{\avvdeg}{\widetilde{\vdeg}} 
\DMO{\cldeg}{cldg} 
\DMO{\mvardu}{\mu\!\vdeg}
\DMO{\Inj}{Inj}
\DMO{\Ex}{Ex} 
\DMO{\surp}{\sigma} 
\DMO{\nonmer}{nM} 
\usepackage{theorem} 
\usepackage[hypertex]{hyperref} 
\newtheorem{defi}{Definition}[section]
\newtheorem{lem}[defi]{Lemma}
\newtheorem{thm}[defi]{Theorem}
\newtheorem{corol}[defi]{Corollary}

\newtheorem{conj}[defi]{Conjecture}

\theorembodyfont{\rmfamily}

\theorembodyfont{}
\nc{\bm}{\boldmath}
\nc{\bmm}[1]{\mbox{\bm$\DST #1$}}
\nc{\mi}[1]{\bmm{\mathrm{(#1):}} \quad}

\DeclareMathOperator{\fld}{fld} 
\newcommand{\Mlcr}{\mc{MLCR}} 

\newenvironment{proof}{\noindent\textbf{Proof} \hspace*{0.2em}}{\Ende}

\begin{document}

\title{On variables with few occurrences in\\ conjunctive normal forms}

\author{
  Oliver Kullmann\\
  Computer Science Department\\
  Swansea University\\
  Swansea, SA2 8PP, UK\\
  O.Kullmann@Swansea.ac.uk\\
  {\small \url{http://cs.swan.ac.uk/~csoliver}}
  \and
  Xishun Zhao\thanks{Supported by NSFC Grant 60970040.}\\
  Institute of Logic and Cognition\\
  Sun Yat-sen University\\
  Guangzhou, 510275, P.R.C.
}

\maketitle

\begin{abstract}
  We consider the question of the existence of variables with few occurrences in boolean conjunctive normal forms (clause-sets). Let $\minvdeg(F)$ for a clause-set $F$ denote the minimal variable-degree, the minimum of the number of occurrences of variables. Our main result is an upper bound $\minvdeg(F) \le \nonmer(\surp(F)) \le \surp(F) + 1 + \log_2(\surp(F))$ for \emph{lean clause-sets} $F$ in dependency on the \emph{surplus} $\surp(F)$. Lean clause-sets, defined as having no non-trivial autarkies, generalise minimally unsatisfiable clause-sets. For the surplus we have $\surp(F) \le \delta(F) = c(F) - n(F)$, using the deficiency $\delta(F)$ of clause-sets, the difference between the number of clauses and the number of variables. $\nonmer(k)$ is the $k$-th ``non-Mersenne'' number, skipping in the sequence of natural numbers all numbers of the form $2^n - 1$. As an application of the upper bound we obtain that clause-sets $F$ violating $\minvdeg(F) \le \nonmer(\surp(F))$ must have a non-trivial autarky (so clauses can be removed satisfiability-equivalently by an assignment satisfying some clauses and not touching the other clauses). It is open whether such an autarky can be found in polynomial time.
\end{abstract}


\section{Introduction}
\label{sec:intro}

We study the existence of ``simple'' variables in boolean conjunctive normal forms, considered as clause-sets. ``Simple'' here means a variable occurring not very often. A major use of the existence of such variables is in inductive proofs of properties of minimally unsatisfiable clause-sets, using splitting on a variable to reduce $n$, the number of variables, to $n-1$: here it is vital that we have control over the changes imposed by the substitution, and so we want to split on a variable occurring as few times as possible. The background for these considerations is the enterprise of classifying minimal unsatisfiable clause-sets $F$ in dependency on the deficiency $\delta(F) := c(F) - n(F)$, the difference between the number $c(F) := \abs{F}$ of clauses of $F$ and the number $n(F) := \abs{\var(F)}$ of variables of $F$. The most basic fact is $\delta(F) \ge 1$, as first shown in \cite{AhLi86}. For $\delta(F) = 1$ the structure is completely known (\cite{AhLi86,DDK98,Ku99dKo}, for $\delta(F) = 2$ the structure after reduction of singular variables (occurring in one sign only once) is known (\cite{KleineBuening2000SubclassesMU}), while for $\delta(F) \in \set{3,4}$ only basic cases have been classified (\cite{XD99}).

The starting point of our investigation is Lemma C.2 in \cite{Ku99dKo}, where it is shown that a minimally unsatisfiable clause-set $F$ must have a variable $v$ with at most $\delta(F)$ positive and at most $\delta(F)$ negative occurrences; we write this as $\ldeg_F(v) \le \delta(F)$ and $\ldeg_F(\ol{v}) \le \delta(F)$, using the notion of \emph{literal degrees} (the number of occurrences of the literal). Thus we have $\vdeg_F(v) \le 2 \delta(F)$, using the \emph{variable degree} $\vdeg_F(v) := \ldeg_F(v) + \ldeg_F(\ol{v})$. Using the \emph{minimum variable degree} (min-var-degree) $\minvdeg(F) := \min_{v \in \var(F)} \vdeg_F(v)$ of $F$, this becomes $\minvdeg(F) \le 2 \delta(F)$. In this article we show a sharper bound on $\minvdeg(F)$ for a larger class of clause-sets $F$. More precisely, we show that the worst-cases $\ldeg_F(v), \ldeg_F(\ol{v}) \le \delta(F)$ can not occur at the same time (for a suitable variable), but actually $\ldeg_F(v) + \ldeg_F(\ol{v}) - \delta(F)$ only grows logarithmically in $\delta(F)$, and this for a larger class of formulas.

The larger class of clause-sets considered is the class $\Lean$ of \emph{lean clause-sets}, which are clause-sets having no non-trivial autarky. For an overview on the theory of minimally unsatisfiable clause-sets and on the theory of autarkies see \cite{Kullmann2007HandbuchMU}. The deficiency $\delta(F) \in \ZZ$ of clause-sets is replaced by the \emph{surplus} $\surp(F) \in \ZZ$, which is the minimal deficiency over all clause-sets $F[V]$ for non-empty variable sets $V \sse \var(F)$, where $F[V]$ is obtained from $F$ by removing clauses which have no variables in $V$, and restricting the remaining clauses to $V$; see \cite{Kullmann2007ClausalFormZI} for more information on the surplus of (generalised) clause-sets. We need to count multiple occurrences of clauses here (which might arise during the process of removing literals with variables not in $V$), and thus actually multi-clause-sets $F$ are used here. Note that by considering $V = \var(F)$ we have $\surp(F) \le \delta(F)$, and by considering $V = \set{v}$ for $v \in \var(F)$ we get $\surp(F) \le \minvdeg(F) - 1$. Now the main result of this article (Theorem \ref{thm:leanminvardeg}) is
\begin{displaymath}
  \minvdeg(F) \le \nonmer(\surp(F))
\end{displaymath}
for lean $F$, where $\nonmer: \NN \ra \NN$ (see Definition \ref{def:minvdegdef}) is a super-linear function with $\nonmer(k) \le k + 1 + \log_2(k)$. As an application we obtain (Corollary \ref{cor:leanminvardegsat}), that if a (multi-)clause-set $F$ has no variable occurring with degree at most $\delta(F) + 1 + \log_2(\delta(F))$, then $F$ has a non-trivial autarky. It is an open problem whether such an autarky can be found in polynomial time (for arbitrary clause-sets $F$); we conjecture (Conjecture \ref{con:findauthard}) that this is possible.

\paragraph{Related work}

This article appears to be the first systematic study of the problem of minimum variable occurrences in minimally unsatisfiable clause-sets and generalisations, in dependency on the deficiency, asking for the existence of a variable occurring ``infrequently'' in general, or for extremal examples where all variables occur not infrequently. The problem of maximum variable occurrences (asking for the existence of a variable occurring frequently in general, or for extremal examples where all variables occur not frequently) in uniform (minimally) unsatisfiable clause-sets, in dependency on the (constant) clause-length, has been studied in the literature, starting with \cite{Tovey1984NPcomplete}; for a recent article see \cite{GebauerSzaboTardos2010LocalLemma}.

\paragraph{Overview}

In Section \ref{sec:prelim} basic notions and concepts regarding clause-sets, autarkies and minimal unsatisfiability are reviewed. Section \ref{sec:nonmer} introduces the numbers $\nonmer(k)$ and proves exact formulas and sharp lower and upper bounds. Section \ref{sec:leansurp} contains the main results. First in Subsection \ref{sec:specialcaseMU} the bound is shown for minimally unsatisfiable clause-sets (Theorem \ref{thm:MUminvdegdef}). In Subsection \ref{sec:proofgencase} the bound then is lifted to lean clause-sets, proving Theorem \ref{thm:leanminvardeg}. The immediate corollary of Theorem \ref{thm:leanminvardeg} is, that if the asserted upper bound on the minimal variable degree is not fulfilled, then a non-trivial autarky must exist (Corollary \ref{cor:leanminvardegsat}). In Subsection \ref{sec:findaut} the problem of finding such autarky is discussed, with Conjecture \ref{con:findauthard} making precise our believe that one can find such autarkies efficiently. In Section \ref{sec:strengthbound} we discuss the sharpness of the bound, and the possibilities to generalise it further. Finally, in Section \ref{sec:open} open problems are stated, culminating in the central Conjecture \ref{con:finhit} about the classification of unsatisfiable hitting clause-sets (or ``disjoint tautologies'' in the terminology of DNFs).

\section{Preliminaries}
\label{sec:prelim}

We follow the general notations and definitions as outlined in \cite{Kullmann2007HandbuchMU}, where also further background on autarkies and minimal unsatisfiability can be found. We use $\NN = \set{1,2,\dots}$ and $\NNZ = \NN \cup \set{0}$.

\subsection{Clause-sets}
\label{sec:prelimcls}

Complementation of literals $x$ is denoted by $\ol{x}$, while for a set $L$ of literals we define $\ol{L} := \set{\ol{x} : x \in L}$. A \textbf{clause} $C$ is a finite and clash-free set of literals (i.e., $C \cap \ol{C} = \es$), while a \textbf{clause-set} is a finite set of clauses. We use $\var(F) := \bc_{C \in F} \var(C)$ for the set of variables of $F$, where $\var(C) := \set{\var(x) : x \in C}$ is the set of variables of clause $C$, while $\var(x)$ is the underlying variable for a literal $x$. For a clause-set $F$ we denote by $n(F) := \abs{\var(F)} \in \NNZ$ the number of variables and by $c(F) := \abs{F} \in \NNZ$ the number of clauses. The \textbf{deficiency} of a clause-set is denoted by $\delta(F) := c(F) - n(F) \in \ZZ$. We call a clause $C$ \textbf{full} for a clause-set $F$ if $\var(C) = \var(F)$, while a clause-set $F$ is called full if every clause is full. For a finite set $V$ of variables let $A(V)$ be the set of all $2^{\abs{V}}$ full clauses over $V$. Thus full clause-sets are exactly the sub-clause-sets of some $A(V)$. A \textbf{partial assignment} is a map $\vp: V \ra \set{0,1}$ for some (possibly empty) set $V$ of variables. The application of a partial assignment $\vp$ to a clause-set $F$ is denoted by $\vp * F$, which yields the clause-set obtained from $F$ by removing all satisfied clauses (which have at least one literal set to $1$), and removing all falsified literals from the remaining clauses. A clause-set $F$ is satisfiable iff there is a partial assignment $\vp$ with $\vp * F = \top := \es$, otherwise $F$ is unsatisfiable. All $A(V)$ are unsatisfiable.

These notions are generalised to \textbf{multi-clause-sets}, which are pairs $(F, m)$, where $F$ is a clause-set and $m: F \ra \NN$ determines the multiplicity of the clauses. Now $c((F,m)) := \sum_{C \in F} m(C)$, while the application of partial assignments $\vp$ to a multi-clause-set $F$ yields a \emph{multi-}clause-set $\vp * F$, where the multiplicity of a clause $C$ in $\vp * F$ is the sum of all multiplicities of clauses in $F$ which are shortened to $C$ by $\vp$. For example if $\vp$ is a total assignment for $F$ (assigns all variables of $F$) which does not satisfying $F$ (i.e., $\vp * F \not= \top$), then $\vp * F$ is $(\set{\bot},(f)_{C \in \set{\bot}})$, where $\bot := \es$ is the empty clause, while $f \in \NN$ is the number of clauses (with their multiplicities) of $F$ falsified by $\vp$.

For the number of occurrences of a literal $x$ in a (multi-)clause-set $(F,m)$ we write $\ldeg_F(x) := \sum_{C \in F, x \in C} m(C)$, called the \textbf{literal-degree}, while the \textbf{variable-degree} of a variable $v$ is defined as $\vdeg_F(v) := \ldeg_F(v) + \ldeg_F(\ol{v})$. A \textbf{singular variable} in a (multi-)clause-set $F$ is a variable occurring in one sign only once (i.e., $1 \in \set{\ldeg_F(v), \ldeg_F(\ol{v})}$). A (multi-)clause-set is called \textbf{non-singular} if it does not have singular variables.

For a set $V$ of variables and a multi-clause-set $F$ by $F[V]$ the \textbf{restriction} of $F$ to $V$ is denoted, which is obtained by removing clauses from $F$ which have no variables in common with $V$, and removing from the remaining clauses all literals where the underlying variable is not in $V$ (note that this can increase multiplicities of clauses).

\subsection{Autarkies}
\label{sec:prelimAut}

An \textbf{autarky} for a clause-set $F$ is a partial assignment $\vp$ which satisfies every clause $C \in F$ it touches, i.e., with $\var(\vp) \cap \var(C) \not= \es$. The empty partial assignment is always an autarky for every $F$, the \textbf{trivial autarky}. If $\vp$ is an autarky for $F$, then $\vp * F \sse F$ holds, and thus $\vp * F$ is satisfiability-equivalent to $F$. A clause-set $F$ is \textbf{lean} if there is no non-trivial autarky for $F$. A weakening is the notion of a \textbf{matching-lean} clause-set $F$, which has no non-trivial \textbf{matching autarky}, which are special autarkies given by a matching condition (for every clause touched, a unique variable underlying a satisfied literal must be selectable). The process of applying autarkies as long as possible to a clause-set is confluent, yielding the \textbf{lean kernel} of a clause-set. Computation of the lean kernel is NP-hard, but the \textbf{matching-lean kernel}, obtained by applying matching autarkies as long as possible, which is also a confluent process, is computable in polynomial time. Note that a clause-set $F$ is lean resp.\ matching lean iff the lean resp.\ matching-lean kernel is $F$ itself. For every matching-lean multi-clause-set $F \not= \top$ we have $\delta(F) \ge 1$, while in general a multi-clause-set $F \not= \top$ is matching lean iff $\surp(F) \ge 1$, where the \textbf{surplus} $\sigma(F) \in \ZZ$ is defined as the minimum of $\delta(F[V])$ for all $\es \not= V \sse \var(F)$. Note that while w.r.t.\ general autarkies there is no difference between a multi-clause-set and the underlying clause-set, for matching autarkies there is a difference, due to the matching condition. For more information on autarkies see \cite{Kullmann2007HandbuchMU,Kullmann2007ClausalFormZI}.

\subsection{Minimally unsatisfiable clause-sets}
\label{sec:MUprelim}

The set of minimally unsatisfiable clause-sets is $\Musat$, the set of all clause-sets which are unsatisfiable, while removal of any clause makes them satisfiable. Furthermore the set of saturated minimally unsatisfiable clause-sets is $\Smusat \subset \Musat$, which is the set of minimally unsatisfiable clause-sets such that addition of any literal to any clause renders them satisfiable. We recall the fact that every minimally unsatisfiable clause-set $F \in \Musat$ can be \textbf{saturated}, i.e., by adding literal occurrences to $F$ we obtain $F' \in \Smusat$ with $\var(F') = \var(F)$ such that there is a bijection $\alpha: F \ra F'$ with $C \sse \alpha(C)$ for all $C \in F$. Some basic properties of $\Musat$ and $\Smusat$ w.r.t.\ the application of partial assignments are given in the following lemma.

\begin{lem}\label{lem:auxSMUSAT}
  For all clause-sets $F$ we have:
  \begin{enumerate}
  \item\label{lem:auxSMUSAT1} $F \in \Smusat$ iff for all $v \in \var(F)$ and $\ve \in \set{0,1}$ we have $\pao{v}{\ve} * F \in \Musat$.
  \item\label{lem:auxSMUSAT2} If for some variable $v$ holds $\pao{v}{0} * F \in \Smusat$ and $\pao{v}{1} * F \in \Smusat$, then $F \in \Smusat$.
  \item\label{lem:auxSMUSAT3} If for some variable $v$ holds $\pao{v}{0} * F \in \Musat$ and $\pao{v}{1} * F \in \Musat$, then $F \in \Musat$.
  \end{enumerate}
\end{lem}
For more information on minimal unsatisfiability see \cite{Kullmann2007HandbuchMU,Kullmann2007ClausalFormZII}.

\section{Non-Mersenne numbers}
\label{sec:nonmer}

Splitting on variables with minimum occurrence in minimally unsatisfiable clause-sets leads by Theorem \ref{thm:MUminvdegdef} to the following recursion. The understanding of this recursion is the topic of this section. On a first reading, only Definition \ref{def:minvdegdef} and the main results, Lemma \ref{lem:solveN} and Corollary \ref{cor:upperboundnonmer}, need to be considered.

\begin{defi}\label{def:minvdegdef}
  For $k \in \NN$ let $\nonmer(k) := 2$ if $k = 1$, while else
  \begin{displaymath}
    \nonmer(k) := \max_{i \in \tb{2}{k}} \min(2 \cdot i, \nonmer(k-i+1) + i).
  \end{displaymath}
\end{defi}
Remarks:
\begin{enumerate}
\item This is sequence \url{http://oeis.org/A062289} in the ``On-Line Encyclopedia of Integer Sequences''. It can be defined as the enumeration of those natural numbers containing the string ``10'' (at consecutive positions). The sequence leaves out exactly the number of the form $2^n-1$ for $n \in \NN$, and thus the name. The sequence consists of arithmetic progressions of slope $1$ and length $2^m - 1$, $m = 1, 2, \dots$, each such progression separated by an additional step of $+1$. The recursion in Definition \ref{def:minvdegdef} is new, and so we can not use these characterisations, but must directly prove the basic properties.
\item The value of $\nonmer(k)$ for $k = (1), (2,3,4), (5,\dots,11)$, $(12,\dots,26)$ is $(2), (4,5,6),\\ (8,\dots,14)$, $(16,\dots,30)$.
\item For $k \ge 2$ we have $\nonmer(k) \ge 4$. This holds since $\nonmer(2) = 4$, while the induction step for $k \ge 3$ is $\nonmer(k) = \max_{i \in \tb 2k} \min(2 i, \nonmer(k-i+1)+i) \ge \min(4, \min(4+2,1+3)) = 4$.
\item By induction and by definition we have $k+1 \le \nonmer(k) \le 2 \cdot k$ for $k \in \NN$.
\end{enumerate}

For a sequence $a: \NN \ra \RR$ and $k \in \NN$ let $\bmm{\Delta a(k)} := a(k+1) - a(k)$ be the step in the value of the sequence from $k$ to $k+1$. The next number in the sequence of non-Mersenne numbers is obtained by adding $1$ or $2$ to the previous number:
\begin{lem}\label{lem:stepNM}
  For $k \in \NN$ holds $\Delta \nonmer(k) \in \set{1,2}$.
\end{lem}
\begin{proof} {\sloppy
  For $k=1$ we get $\Delta \nonmer(1) = 2$. Now consider $k \ge 2$. We have $\nonmer(k+1) = \max(\min(4,\nonmer(k)+2), \max_{i \in \tb 3{k+1}} \min(2 i, \nonmer(k-i+2)+i)) =
  \max_{i \in \tb 3{k+1}} \min(2 i, \nonmer(k-i+2)+i) = \max_{i \in \tb 2k} \min(2 (i+1), \nonmer(k-(i+1)+2)+(i+1)) =
  \max_{i \in \tb 2k} \min(2 i + 2, \nonmer(k-i+1)+i+1) = 1 + \max_{i \in \tb 2k} \min(2 i + 1, \nonmer(k-i+1)+i)$.

Thus on the one hand we have $\nonmer(k+1) \ge 1 + \max_{i \in \tb 2k} \min(2 i, \nonmer(k-i+1)+i) = 1 + \nonmer(k)$, and on the other hand $\nonmer(k+1) \le 1 + \max_{i \in \tb 2k} \min(2 i + 1, \nonmer(k-i+1)+i+1) = 2 + \nonmer(k)$.
}
\end{proof}

\begin{corol}\label{cor:NMmon}
  $\nonmer: \NN \ra \NN$ is strictly increasing.
\end{corol}

\begin{corol}\label{cor:NMdiff}
  We have $\nonmer(a+b) \ge \nonmer(a) + b$ for $a \in \NN$ and $b \in \NNZ$, and thus $\nonmer(a - b) \le \nonmer(a) - b$ for $b \le a$.
\end{corol}

Instead of considering the maximum over $k-1$ cases $i \in \tb 2k$ to compute $\nonmer(k)$, we can now simplify the recursion to only one case $i(k) \in \tb 2k$, and for that case also consideration of the minimum is dispensable:
\begin{lem}\label{lem:simprecnm}
  For $k \in \NN$, $k \ge 2$, let $i(k) \in \NN$ be the smallest $i \in \tb 2k$ with $i \ge \nonmer(k-i+1)$ (note that $k \ge \nonmer(k-k+1) = 2$, and thus $i(k)$ is well-defined). For example we have $i(2) = 2$, $i(3) = 3$, $i(4) = 4$ and $i(5) = 4$. Then we have:
  \begin{enumerate}
  \item\label{lem:simprecnm1} $i(k) - \nonmer(k-i(k)+1) \le 2$.
  \item\label{lem:simprecnm2} $\nonmer(k) = \nonmer(k-i(k)+1)+i(k)$.
  \item\label{lem:simprecnm3} $\Delta i(k) \in \set{0,1}$.
  \end{enumerate}
\end{lem}
\begin{proof} We have $i(k) = 2$ iff $k=2$, while for $k=2$ the assertions hold trivially; so assume $k \ge 3$ and $i(k) \ge 3$. Part \ref{lem:simprecnm1} follows by Lemma \ref{lem:stepNM} from the facts that the sequence $i \in \tb 2k \mapsto i$ moves up in steps of $+1$, while the sequence $i \in \tb 2k \mapsto \nonmer(k-i+1)$ moves down in steps of $-1$ or $-2$. It remains to show Part \ref{lem:simprecnm2}. By Lemma \ref{lem:stepNM} the sequence $i \in \tb 2k \mapsto \nonmer(k-i+1)+i$ is monotonically decreasing, and thus by definition we obtain $\nonmer(k) = \max(2 \cdot (i(k)-1), \nonmer(k-i(k)+1)+i(k))$. That the maximum here is actually always attained in the second component follows by Part \ref{lem:simprecnm1}. Finally Part \ref{lem:simprecnm3} follows again from Lemma \ref{lem:stepNM}.
\end{proof}

After these preparations we are able to characterise the ``jump positions'', the set $J \subset \NN$ of $k \in \NN$ with $\Delta \nonmer(k) = 2$. Thus $\Delta \nonmer(k) = 1$ iff $k \notin J$, and $J = \set{1,4,11,26,\dots}$. Note $\nonmer(k) = 1 + k + \abs{\set{k' \in J : k' < k}}$.
\begin{lem}\label{lem:characjump}
  Let $i'(k) := k-i(k)+1$ and $h(k) := \nonmer(i'(k))$ for $k \in \NN$, $k \ge 2$. Thus $\Delta i'(k) \in \set{0,1}$ and $\Delta i(k) = 1 - \Delta i'(k)$. Furthermore we have $\nonmer(k) = h(k) + i(k)$, thus $\Delta \nonmer(k) = \Delta h(k) + \Delta i(k)$, and $i(k) - h(k) \in \set{0,1,2}$. Consider $k \ge 2$.
  \begin{enumerate}
  \item\label{lem:characjump1} If $\Delta i(k) = 0$, then:
    \begin{enumerate}
    \item\label{lem:characjump1a} $\Delta i(k+1) = 1$
    \item\label{lem:characjump1b} $i(k) \not= h(k)$.
    \item\label{lem:characjump1c} $i(k+1) = h(k+1)$.
    \end{enumerate}
  \item\label{lem:characjump2} If $\Delta i(k) = 1$, then:
    \begin{enumerate}
    \item\label{lem:characjump2a} $\Delta h(k) = 0$, and so $k \notin J$
    \item\label{lem:characjump2b} $i(k) \not= h(k)+2$.
    \end{enumerate}
  \item\label{lem:characjump3} The following conditions are equivalent:
    \begin{enumerate}
    \item\label{lem:characjump3a} $k \in J$
    \item\label{lem:characjump3c} $\Delta h(k) = 2$
    \item\label{lem:characjump3d} $i(k) = h(k) + 2$
    \item\label{lem:characjump3g} $\Delta i(k-1) = 1$ and $i(k-1) = h(k-1) + 1$
    \item\label{lem:characjump3e} $\Delta i(k-2) = \Delta i(k-1) = 1$
    \item\label{lem:characjump3b} $i'(k) = i'(k-1) = i'(k-2)$ and $i'(k) \in J$.
    \end{enumerate}
  \item\label{lem:characjump4} If $k \in J$, then $i'(k) = \max({k' \in J : k' < k})$.
  \end{enumerate}
\end{lem}
\begin{proof} Part \ref{lem:characjump1a} follows by definition. For Part \ref{lem:characjump1b} note $i(k+1) = i(k)$ while $h(k+1) \ge h(k) + 1$. For Part \ref{lem:characjump1c} assume $i(k+1) > h(k+1)$. Then we have $i(k) = h(k)+2$ and $h(k+1) = h(k)+1$. However then $i(k)-1 = h(k)+1 = h(k+1) = \nonmer(k - (i(k)-1) + 1)$ contradicting the definition of $i(k)$. For Part \ref{lem:characjump2a} assume $i(k) = i(k+1) = i(k+2)$. We have $i(k) \ge h(k+2) = \nonmer(k-i(k)+3)$, while $i(k) - 1 < \nonmer(k-(i(k)-1)+1) = \nonmer(k-i(k)+2)$, i.e., $i(k) \le \nonmer(k-i(k)+2)$, contradicting the strict monotonicity of $\nonmer$. Part \ref{lem:characjump2b} follows by $i(k+1) \le h(k+1) + 2$ and $i(k+1) = i(k)+1$, $h(k+1) = h(k)$. Now consider Part \ref{lem:characjump3}.

Condition \ref{lem:characjump3a} implies condition \ref{lem:characjump3c} due to $\Delta i(k) = 0$ in case of $k \in J$ by Part \ref{lem:characjump2a}. Condition \ref{lem:characjump3c} implies condition \ref{lem:characjump3d}, since $\Delta h(k) = 2$ implies $\Delta i(k) = 0$ (otherwise we had $\Delta \nonmer(k) = 3$), and so by Part \ref{lem:characjump1c} we have $i(k) = i(k+1) = h(k+1)$, while the assumption says $h(k+1) = h(k) + 2$. In turn condition \ref{lem:characjump3d} implies condition \ref{lem:characjump3a}, since by Part \ref{lem:characjump2b} we get $\Delta i(k) = 0$, and thus $\Delta \nonmer(k) = \Delta h(k)$, while in case of $\Delta h(k) \le 1$ we would have $i(k)-1 \ge \nonmer(k-(i(k)-1)+1)$ contradicting the definition of $i(k)$, due to $\nonmer(k-(i(k)-1) + 1) = \nonmer((k+1) - i(k+1) + 1) = h(k+1) \le h(k) + 1 = i(k) - 1$. So now we can freely use the equivalence of these three conditions.

Condition \ref{lem:characjump3d} implies condition \ref{lem:characjump3g}, since we have $\Delta i(k) = 0$, and thus $\Delta i(k-1) = 1$ with Part \ref{lem:characjump1a}, from which we furthermore get $i(k) = i(k-1) + 1$ and $h(k-1) = h(k)$, and so $i(k-1) = i(k) - 1 = h(k) + 1 = h(k-1) + 1$. Condition \ref{lem:characjump3g} implies condition \ref{lem:characjump3e}, since in case of $\Delta i(k-2) = 0$ we had $i(k-1) = h(k-1)$ with Part \ref{lem:characjump1c}. In turn condition \ref{lem:characjump3e} implies condition \ref{lem:characjump3d}, since $i(k) = i(k-1) + 1 = i(k-2) + 2$, while $h(k) = h(k-1) = h(k-2)$, where by definition $i(k-2) \ge h(k-2)$ holds, whence $i(k) \ge h(k) + 2$, which implies $i(k) = h(k) + 2$. So now the first five conditions have been shown to be equivalent.

Now condition \ref{lem:characjump3e} implies condition \ref{lem:characjump3b}, since it only remains to show $i'(k) \in J$, which follows with condition \ref{lem:characjump3c} (using $\Delta i(k) = 0$). In turn condition \ref{lem:characjump3b} implies immediately condition \ref{lem:characjump3e}.

Finally, we prove Part \ref{lem:characjump4} by induction on $k$ (regarding the enumeration of $J$). We have $i'(4) = 1$, and so the induction holds for $k=4$, the smallest jump position $k \ge 2$. Now assume that the assertion holds for all elements of $J \cap \tb 1{k-1}$, where $k > 4$, and we have to show the assertion for $k$. By Part \ref{lem:characjump3b} we know $i'(k) \in J$, where $2 \le i'(k) < k$. Assume there is $k' \in J$ with $i'(k) < k' < k$. Now by induction hypothesis we get $i'(k) \le i'(k') < k'$. However by Part \ref{lem:characjump1} we get $\Delta i'(k') = 1$, and thus $i'(k) > i'(k')$ (since $k > k'$).
\end{proof}

\begin{corol}\label{cor:CharacJ}
  We have $J = \set{2^{m+1} - m - 2 : m \in \NN}$.
\end{corol}
\begin{proof} Let $k_m$ for $m \in \NN$ be the $m$th element of $J$; so the assertion is $k_m = 2^{m+1} - m - 2$. We have $k_1 = 4 - 1 - 2 = 1 = \min J$; in the remainder assume $m \ge 2$. We prove the assertion by induction, in parallel with $i(k_m) = 2^{m+1} - 2^m$. For $m=2$ we have $k_2 = 8 - 2 - 2 = 4 = \min J \sm \set{1}$, while $i(4)$ is the smallest $i \in \set{2,3,4}$ with $i \ge \nonmer(5-i)$, which yields $i(4) = 4 = 2^3 - 2^2$. Now we consider the induction step, from $m-1$ to $m$. The induction hypothesis yields $k_{m-1} = 2^m - m - 1$ and $i(k_{m-1}) = 2^m - 2^{m-1}$. Lemma \ref{lem:characjump}, Part \ref{lem:characjump4} yields $i'(k_m) = k_{m-1}$, from which by $i'(k_m) = k_m - i(k_m) + 1$ follows $k_m = 2^m - m - 2 + i(k_m)$. By definition we get $i(k_m) = \Delta i(k_m-1) + \dots + \Delta i(k_{m-1}) + i(k_{m-1})$. By Lemma \ref{lem:characjump}, Parts \ref{lem:characjump1} - \ref{lem:characjump3} the sequence of $\Delta$-values has the form (starting with the lowest index) $0,1,\: 0,1,\: \dots, 0,1,1$, and thus their sum has the value $\frac 12 (k_m - k_{m-1} - 1) + 1$. So we get $i(k_m) = \frac 12 (k_m - k_{m-1} - 1) + 1 + i(k_{m-1}) = \frac 12 (2^m - m - 2 + i(k_m) - 2^m + m + 1 - 1) + 1 + 2^m - 2^{m-1} = \frac 12 i(k_m) - 1 + 1 + 2^m - 2^{m-1}$, from which $i(k_m) = 2^{m+1} - 2^m$ follows. Finally $k_m = 2^m - m - 2 + 2^{m+1} - 2^m = 2^{m+1} - m - 2$.
\end{proof}

Now the closed formula for $\nonmer(k)$ can be proven (using $\ld(x) := \log_2(x)$):
\begin{lem}\label{lem:solveN}
  For $k \in \NN$ let $\fld(k) := \floor{\ld(k)}$ (``floor of logarithm dualis''). Then we have for $k \in \NN$ the equality $\nonmer(k) = k + \fld(k+1 + \fld(k+1))$.
\end{lem}
\begin{proof} Let $g(k) := \fld(k+1 + \fld(k+1))$ and $f(k) := k + g(k)$ (so $\nonmer(k) = f(k)$ is to be shown, for $k \ge 1$). We have $f(1) = 1 + \fld(2 + \fld(2)) = 1 + \fld(3) = 2 = \nonmer(1)$. We will now prove that the function $g(k)$ changes values exactly at the transitions $k \mapsto k+1$ for $k \in J$, that is, for indices $k = k_m := 2^{m+1} - m - 2$ (using Corollary \ref{cor:CharacJ}) with $m \in \NN$ we have $\Delta g(k_m) = 1$, while otherwise we have $\Delta g(k_m) = 0$, from which the assertion follows (by the definition of $J$).

We have $g(1) = 1$ and $g(2) = 2$. Now consider $m \in \NN$ and $k_m + 1 \le k \le k_{m+1}$. We show $g(k) = m+1$, which proves the claim. Note that $g(k)$ is monotonically increasing. Now $g(k) \ge g(k_m+1) = \floor{\ld(2^{m+1}-m + \floor{\ld(2^{m+1}-m)})} = \floor{\ld(2^{m+1}-m + m)} = m+1$ and $g(k) \le g(k_{m+1}) = \floor{\ld(2^{m+2}-m-2 + \floor{\ld(2^{m+2}-m-2)})} \le \floor{\ld(2^{m+2}-m-2 + m+1)} = \floor{\ld(2^{m+2}-1)} = m+1$.
\end{proof}

As a result, we obtain very precise bounds:
\begin{corol}\label{cor:upperboundnonmer}
  $k + \fld(k+1) \le \nonmer(k) \le k+1+\fld(k)$ holds for $k \in \NN$.
\end{corol}
\begin{proof} The lower bound follows trivially. The upper bound holds (with equality) for $k \le 2$, so assume $k \ge 3$. We have to show $g(k) = \fld(k+1+\fld(k+1)) \le 1 + \fld(k)$, which follows from $\ld(k+1+ \fld(k+1)) \le 1 + \ld(k)$. Now $\ld(k+1+ \fld(k+1)) \le \ld(k+1+ \ld(k+1)) \le \ld(k+k) = 1 + \ld(k)$.
\end{proof}

\section{Lean clause-sets and the surplus}
\label{sec:leansurp}

In this section we prove the main result of this paper, Theorem \ref{thm:leanminvardeg}. The proof consists in first handling a special case, minimally unsatisfiable clause-sets instead of lean clause-sets, in Subsection \ref{sec:specialcaseMU}, and then lifting the result to the general case in Subsection \ref{sec:proofgencase}. In Subsection \ref{sec:findaut} we consider the algorithmic implications of this result.

\begin{thm}\label{thm:leanminvardeg}
  We have $\minvdeg(F) \le \nonmer(\surp(F))$ for a lean multi-clause-set $F$ with $n(F) > 0$. More precisely, there exists a variable $v \in \var(F)$ with $\vdeg_F(v) \le \nonmer(\surp(F))$ and $\ldeg_F(v), \ldeg_F(\ol{v}) \le \surp(F)$.
\end{thm}

We obtain a sufficient criterion for the existence of a non-trivial autarky.
\begin{corol}\label{cor:leanminvardegsat}
  Consider a multi-clause-set $F$ with $n(F) > 0$. If $\surp(F) \le 0$, then $F$ has a non-trivial matching autarky. So assume $\surp(F) \ge 1$. If we have $\minvdeg(F) > \nonmer(\surp(F))$, then for every $\es \not= V \sse \var(F)$ with $\delta(F[V]) = \surp(F)$ we have an autarky $\vp$ for $F$ with $\var(\vp) = V$ (and thus $F$ has a non-trivial autarky).
\end{corol}

The quantities $\minvdeg(F)$ and $\nonmer(\surp(F))$ (resp.\ $\nonmer(\delta(F))$) are computable in polynomial time, and so the applicability of Corollary \ref{sec:proofgencase} can be checked in polynomial time. We conjecture that also ``constructivisation'' of Corollary \ref{cor:leanminvardegsat} can be done in polynomial time:
\begin{conj}\label{con:findauthard}
   There is a poly-time algorithm for computing a non-trivial autarky in case of $\minvdeg(F) > \nonmer(\surp(F))$ (or $\minvdeg(F) > \nonmer(\delta(F))$) for matching-lean clause-sets $F$.
\end{conj}
See Subsection \ref{sec:findaut} for more discussion on Conjecture \ref{con:findauthard} (there also the remaining details of Corollary \ref{cor:leanminvardegsat} are proven).

\subsection{The special case of minimally unsatisfiable clause-sets}
\label{sec:specialcaseMU}

The main auxiliary lemma is the following statement, which receives its importance from the fact that every minimally unsatisfiable clause-set can be saturated (this method was first applied in \cite{Ku99dKo}).
\begin{lem}\label{lem:auxminvardeg}
  Consider $F \in \Smusati{\delta=k}$ for $k \in \NN$ and a variable $v \in \var(F)$ realising the minimal var-degree (i.e., $\vdeg_F(v) = \minvdeg(F)$). Using $m_0 := \ldeg_F(\ol{v})$ and $m_1 := \ldeg_F(v)$ we have $\pao v{\ve} * F \in \Musati{k-m_{\ve}+1}$ for $\ve \in \set{0,1}$, where $n(\pao v{\ve} * F) = n(F) - 1$. Since minimally unsatisfiable clause-sets have deficiency at least one, we get $m_{\ve} \le k$.
\end{lem}
\begin{proof} We have $n(\pao v{\ve} * F) = n(F) - 1$ since $F$ contains no pure variable, while $v$ realises the minimum of var-degrees. Thus $\delta(\pao v{\ve} * F) = \delta(F) - m_{\ve} + 1$, while $\pao v{\ve} * F \in \Musat$ by Lemma \ref{lem:auxSMUSAT}, Part \ref{lem:auxSMUSAT1}.
\end{proof}

\begin{thm}\label{thm:MUminvdegdef}
  For all $k \in \NN$ and $F \in \Musati{\delta \le k}$ we have $\minvdeg(F) \le \nonmer(k)$. More precisely, for $n(F) > 0$ there exists a variable $v \in \var(F)$ with $\vdeg_F(v) \le \nonmer(k)$ and $\ldeg_F(v), \ldeg_F(\ol{v}) \le k$.
\end{thm}
\begin{proof} The assertion is known for $k=1$, so assume $k > 1$, and we apply induction on $k$. Assume $\delta(F) = k$ (due to $k > 1$ we have $n(F) > 1$). Saturate $F$ and obtain $F'$. Consider a variable $v \in \var(F')$ realising the min-var-degree of $F'$. If $\vdeg_{F'}(v) = 2$ then we are done, so assume $\vdeg_{F'}(v) \ge 3$. Let $i := \max(\ldeg_{F'}(v),\ldeg_{F'}(\ol{v}))$; so $\vdeg_{F'}(v) \le 2 i$. W.l.o.g.\ assume that $i = \ldeg_{F'}(v)$. By Lemma \ref{lem:auxminvardeg} we get $2 \le i \le k$. Applying the induction hypothesis and Lemma \ref{lem:auxminvardeg} we obtain a variable $w \in \var(G)$ for $G := \pao v1 * F$ with $\vdeg_G(w) \le \nonmer(k-i+1)$. By definition we have $\vdeg_{F'}(w) \le \vdeg_G(w) + \ldeg_{F'}(v)$. Altogether we get $\minvdeg(F) \le \min(2 i, \nonmer(k-i+1) + i) \le \nonmer(k)$.
\end{proof}

It is interesting to generalise Theorem \ref{thm:MUminvdegdef} for generalised clause-sets (see \cite{Kullmann2007ClausalFormZI,Kullmann2007ClausalFormZII} for a systematic study, and \cite{Kullmann2011ClausalForm} for the underlying report). Generalised clause-sets have literals ``$v \not= \ve$'' for variables $v$ with domains $D_v$ and values $\ve \in D_v$, and the deficiency is generalised by giving every variable a weight $\abs{D_v} - 1$ (which is $1$ in the boolean case). The base case of deficiency $k=1$ is handled in Lemma 5.4 in \cite{Kullmann2007ClausalFormZII}, showing that for generalised clause-sets we have here $\minvdeg(F) \le \max_{v \in \var(F)} \abs{D_v}$. But $k \ge 2$ requires more work:
\begin{enumerate}
\item The basic method of saturation is not available for generalised clause-sets, as discussed in Subsection 5.1 in \cite{Kullmann2007ClausalFormZII}. Thus the proofs for the boolean case seem not to be generalisable.
\item Stipulating the effects of saturation via the ``substitution stability parameter regarding irredundancy'', in Corollary 5.10 in \cite{Kullmann2007ClausalFormZII} one finds a first approach towards generalising the basic bound $\minvdeg(F) \le 2 \delta(F)$ (for the boolean case) by $\minvdeg(F) \le \max_{v \in \var(F)} \abs{D_v} \cdot \delta(F)$.
\item Another approach uses translations to boolean clause-sets. The ``generic translation scheme'' (see \cite{Kullmann2010GreenTao,Kullmann2007ClausalFormZII}) allows (for certain instances) to preserve the deficiency and the other structures relevant here. So we get general upper bounds for the minimum number of occurrences of variables in generalised clause-sets from the boolean case. But further investigations are needed in these bounds.
\end{enumerate}

\subsection{Proof of the general case}
\label{sec:proofgencase}

Now consider an arbitrary (multi-)clause-set $F$. Consider a set of variables $\es \not= V \sse \var(F)$ realising the surplus of $F$, i.e., such that $\delta(F[V])$ is minimal. If $F[V]$ would be satisfiable, then a satisfying assignment would give a non-trivial autarky for $F$. Assuming that $F$ is lean thus yields that $F[V]$ must be unsatisfiable. So there exists a minimally unsatisfiable $F' \sse F[V]$. If now $\var(F') \not= \var(F[V]) = V$ would be the case, then we would loose control over the deficiency of $F'$. Fortunately this can not happen, as the following lemma shows.

\begin{lem}\label{lem:auxminvardegsigma}
  Consider a multi-clause-set $F$ with $\surp(F) = \delta(F)$. Then for every unsatisfiable sub-multi-clause-set $F' \le F$ we have $\var(F') = \var(F)$.
\end{lem}
\begin{proof} Assume $\var(F') \subset \var(F)$, and consider a minimally unsatisfiable sub-clause-set $F'' \sse F'$. By definition we have $\delta(F'') + \delta(F[\var(F) \sm \var(F'')]) \le \delta(F)$, where $\delta(F[\var(F) \sm \var(F'')]) \ge \surp(F) = \delta(F)$, from which we conclude $\delta(F'') \le 0$, but $\delta(F'') \ge 1$ must hold since $F''$ is minimally unsatisfiable.
\end{proof}

Finally we are able to prove Theorem \ref{thm:leanminvardeg}. Recall that $F$ is a lean multi-clause-set with $n(F) > 0$, and we have to show the existence of a variable $v$ with $\vdeg_F(v) \le \nonmer(\surp(F))$ and $\ldeg_F(v), \ldeg_F(\ol{v}) \le \surp(F)$.

Consider $\es \not= V \sse \var(F)$ with $\delta(F[V]) = \surp(F)$, and let $F' := F[V]$. $F'$ is unsatisfiable, since $F$ is lean. Because of $\delta(F') = \surp(F)$ we have $\delta(F') = \surp(F')$. Consider some minimally unsatisfiable $F'' \sse F'$. By Lemma \ref{lem:auxminvardegsigma} we have $\var(F'') = \var(F')$. So we get $\delta(F'') = \delta(F') - (c(F') - c(F''))$. By Theorem \ref{thm:MUminvdegdef} there is $v \in \var(F'')$ with $\vdeg_{F''}(v) \le \nonmer(\delta(F'')) = \nonmer(\delta(F') - (c(F') - c(F''))) \le \nonmer(\delta(F')) - (c(F') - c(F''))$ and $\ldeg_{F''}(v), \ldeg_{F''}(\ol{v}) \le \delta(F'') = \delta(F') - (c(F') - c(F''))$. Finally we have $\vdeg_F(v) \le \vdeg_{F''}(v) + (c(F') - c(F''))$ (note that all occurrences of $v$ in $F$ are also in $F'$), and similarly for the literal degrees. QED

\begin{corol}\label{cor:leanminvardegdef}
  For a lean multi-clause-set $F$ with $n(F) > 0$ we have $\minvdeg(F) \le \nonmer(\delta(F))$.
\end{corol}

\begin{corol}\label{cor:charsurp1}
  Consider a lean multi-clause-set $F$.
  \begin{enumerate}
  \item\label{cor:charsurp1a} $\surp(F) = 1$ holds if and only if $\minvdeg(F) = 2$ holds.
  \item\label{cor:charsurp1b} $\minvdeg(F) = 3$ implies $\surp(F) = 2$.
  \end{enumerate}
\end{corol}
\begin{proof} First consider Part \ref{cor:charsurp1a}. If $\surp(F) = 1$ (so $n(F) > 0$), then by Theorem \ref{thm:leanminvardeg} we have $\minvdeg(F) \le \nonmer(1) = 2$, while in case of $\minvdeg(F) = 1$ there would be a matching autarky for $F$. If on the other hand $\minvdeg(F) = 2$ holds, then by definition $\surp(F) \le 2 - 1 = 1$, while $\surp(F) \ge 1$ holds since $F$ is matching lean. For Part \ref{cor:charsurp1b} note that due to $\surp(F)+1 \le \minvdeg(F)$ we have $\surp(F) \le 2$, and then the assertion follows by Part \ref{cor:charsurp1a}.
\end{proof}

Remarks:
\begin{enumerate}
\item If $F$ is lean, then $\surp(F) = 2$ implies $\minvdeg(F) \in \set{3,4}$. An example for $\minvdeg(F) = 4$ is given by the full unsatisfiable clause-set with $2$ variables.
\item Is there a minimally unsatisfiable $F$ with $\minvdeg(F) = 4$ and $\surp(F)=3$?
\item More generally, is there for every $k \in \NN$ a minimally unsatisfiable $F$ with $\sigma(F) = k$ and $\minvdeg(F) = k+1$?
\end{enumerate}

\subsection{On finding the autarky}
\label{sec:findaut}

The following lemma (with Theorem \ref{thm:leanminvardeg}) yields the proof of Corollary \ref{cor:leanminvardegsat}:
\begin{lem}\label{lem:charakappcor}
  Consider a matching-lean multi-clause-set $F$ with $n(F) > 0$. If we have $\minvdeg(F) > \nonmer(\surp(F))$, then all $F[V]$ for $\es \subset V \sse \var(F)$ with $\delta(F[V]) = \surp(F)$ are satisfiable.
\end{lem}
\begin{proof} If some $F[V]$ would be unsatisfiable, then by the proof of Theorem \ref{thm:leanminvardeg} in Subsection \ref{sec:proofgencase} there would be a variable $v$ with $\vdeg_F(v) \le \nonmer(\surp(F))$.
\end{proof}

Now consider a matching-lean multi-clause-set $F$ with $n(F) > 0$, where Corollary \ref{cor:leanminvardegsat} is applicable (recall that we have $\surp(F) \ge 1$), that is, we have $\minvdeg(F) > \nonmer(\surp(F))$. So we know that $F$ has a non-trivial autarky. Conjecture \ref{con:findauthard} states that finding such a non-trivial autarky in this case can be done in polynomial time (recall that finding a non-trivial autarky in general is NP-complete, which was shown in \cite{Ku00f}).

The task of actually finding the autarky can be considered as finding a satisfying assignment for the following class $\Mlcr \subset \Sat \cap \Mlean$ of satisfiable(!) clause-sets $F$, obtained by considering all $F[V]$ for minimal sets of variables $V$ with $\delta(F[V]) = \surp(F)$ (where ``CR'' stands for ``critical''):
\begin{defi}\label{def:Mlcr}
  Let \bmm{\Mlcr} be the class of clause-sets $F$ fulfilling the following three conditions:
  \begin{enumerate}
  \item\label{def:Mlcr1} $F$ is matching-lean, has at least one variable, and does not contain the empty clause.
  \item\label{def:Mlcr2} The only $\es \not= V \sse \var(F)$ with $\delta(F[V]) = \surp(F)$ is $V = \var(F)$ (and thus we have $\delta(F) = \surp(F)$).
  \item\label{def:Mlcr3} $\minvdeg(F) > \nonmer(\surp(F))$.
  \end{enumerate}
\end{defi}
It is sufficient to find a non-trivial autarky for this class of satisfiable clause-sets.
\begin{lem}\label{lem:HardMlcr}
  Conjecture \ref{con:findauthard} is equivalent to the statement, that finding a non-trivial autarky for clause-sets in $\Mlcr$ can be achieved in polynomial time.
\end{lem}
At the time of writing this article, we are not aware of elements of $\Mlcr$ with a deficiency at least $2$.

\section{On strengthening the bound}
\label{sec:strengthbound}

For a class $\mc{C}$ of clause-sets let $\minvdeg(\mc{C})$ be the supremum of $\minvdeg(F)$ for $F \in \mc{C}$ with $n(F) > 0$. So by Theorem \ref{thm:MUminvdegdef} we have $\minvdeg(\Musati{\delta=k}) \le \nonmer(k)$ for all $k \in \NN$. The task of precisely determining $\minvdeg(\Musati{\delta=k})$ for all $k$ will be pursued in the forthcoming \cite{KullmannZhao2010Extremal}; we need more theory for minimally unsatisfiable clause-sets (especially for unsatisfiable hitting clause-sets), and so here we can only mention some results connected with this article.
\begin{itemize}
\item We can show for infinitely many $k$ that $\minvdeg(\Musati{\delta=k}) = \nonmer(k)$.
\item We can also show that the smallest $k$ where we don't have equality is $k=6$, namely $\minvdeg(\Musati{\delta=6}) = 8 = \nonmer(6) - 1$.
\item Let $\nonmer_1: \NN \ra \NN$ be defined by the recursion as in Definition \ref{def:minvdegdef}, however with different start values, namely $\nonmer_1(k) := \nonmer(k)$ for $1 \le k \le 5$, while $\nonmer_1(6) := \nonmer(6) - 1 = 8$. We have $\nonmer_1(k) = \nonmer(k)$ for $k \notin \set{2^m-m+1 : m \in \NN, m \ge 3}$, while for $k = 2^m-m+1$ we have $\nonmer_1(k) = \nonmer(k) - 1 = 2^m$.
\item With the same proof as for Theorem \ref{thm:MUminvdegdef} we can show $\minvdeg(\Musati{\delta=k}) \le \nonmer_1(k)$ for all $k \in \NN$.
\item It seems that this bound can not be generalised to lean clause-sets (as in Theorem \ref{thm:leanminvardeg}).
\end{itemize}

\begin{conj}\label{con:sharpness}
  For all $k \in \NN$ we have $\minvdeg(\Musati{\delta=k}) \ge \nonmer(k)-1$.
\end{conj}

Now we consider the question whether the bound holds for a larger class of clause-sets, that is, whether Theorem \ref{thm:leanminvardeg} can be generalised further, incorporating non-lean clause-sets. We consider the large class $\Mlean$ of matching lean clause-sets, as introduced in \cite{Ku00f}, which is natural, since a basic property of $F \in \Musat$ used in the proof of Theorem \ref{thm:leanminvardeg} is $\delta(F) \ge 1$ for $F \not= \top$, and this actually holds for all $F \in \Mlean$. We will construct for arbitrary deficiency $k \in \NN$ and $K \in \NN$ clause-sets $F \in \Mlean$ of deficiency $k$ where every variable occurs positively at least $K$ times. Thus neither the upper bound $\max(\ldeg_F(v), \ldeg_F(\ol{v})) \le f(\delta(F))$ nor $\ldeg_F(v) + \ldeg_F(\ol{v}) = \vdeg_F(v) \le f(\delta(F))$ for some chosen variable $v$ and for any function $f$ does hold for $\Mlean$. 

An example for $F \in \Mleani{\delta=1}$ with $\minldeg(F) \ge 2$ (and thus $\minvdeg(F) \ge 4$) is given in Section 5 in \cite{Ku2003b}, displaying a ``star-free'' (thus satisfiable) clause-set $F$ with deficiency $1$. In Subsection 9.3 in \cite{Kullmann2007ClausalFormZI} it is shown that this clause-set is matching lean. ``Star-freeness'' in our context means, that there are no singular variables (occurring in one sign only once). Our simpler construction pushes the number of positive occurrences arbitrary high, but there are variables with only one negative occurrence (i.e., there are singular variables).

For a finite set $V$ of variables let $M(V) \sse A(V)$ be the full clause-set over $V$ containing all full clauses with at most one complementation. Obviously $\delta(F) = 1$ holds, and it is easy to see that $M(V) \in \Mlean$ (for every $\es \not= F' \subset F \sse A(V)$ we have $\delta(F') < \delta(F)$, and thus a full clause-set $F$ is matching lean iff $\delta(F) \ge 1$). Furthermore by definition we have $\ldeg_{M(V)}(v) = \abs{V}$ and $\ldeg_{M(V)}(\ol{v}) = 1$ for $v \in V$.

\begin{lem}\label{lem:exmpmlean}
  For $k \in \NN$ and $K \in \NN$ there are clause-sets $F \in \Mleani{\delta=k}$ such that for all variables $v \in \var(F)$ we have $\ldeg_F(v) \ge K$.
\end{lem}
\begin{proof} For $k = 1$ we can set $F := M(\tb {v_1}{v_K})$; so assume $k \ge 2$. Consider any clause-set $G \in \Mleani{\delta=k-1}$ with $n := n(G) \ge K$ (for example we could use $F \in \Musati{\delta=k-1}$), and let $V := \var(G)$. Consider a disjoint copy of $V$, that is a set $V'$ of variables with $V' \cap V = \es$ and $\abs{V'} = \abs{V}$, and consider two enumerations of the clauses $M(V) = \set{C_1, \dots, C_{n+1}}$, $M(V') = \set{C_1', \dots, C_{n+1}'}$. Now
\begin{displaymath}
  F := G \cup \setb{ C_i \cup C_i' : i \in \tb{1}{n+1}}
\end{displaymath}
has no matching autarky: If $\vp$ is a matching autarky for $F$, then $\var(\vp) \cap V = \es$ since $G$ is matching lean, whence $\var(\vp) \cap V' = \es$ since $M(V')$ is matching lean, and thus $\vp$ must be trivial. Furthermore we have $n(F) = 2 n$ and $c(F) = c(G) + n + 1$, and thus $\delta(F) = c(G) + n + 1 - 2 n = \delta(G)+1 = k$. By definition for all variables $v \in \var(F)$ we have $\ldeg_F(v) \ge n$.
\end{proof}

Remarks:
\begin{enumerate}
\item It remains open whether for deficiency $k \in \NN$ we find examples $F \in \Mleani{\delta=k}$ with $\minldeg(F) \ge k+1$ (the above mentioned star-free clause-sets shows that this is the case for $k=1$), or stronger, $\minldeg(F) \ge K$ for arbitrary $K \in \NN$.
\item The clause-sets $F$ constructed in  Lemma \ref{lem:exmpmlean} are not elements of $\Mlcr_{\delta=k}$ for $k \ge 2$, since $\delta(F[V']) = n+1 - n = 1$, thus $\surp(F) = 1$, and so Condition \ref{def:Mlcr2} of Definition \ref{def:Mlcr} is not fulfilled. The corresponding autarky is a satisfying assignment of $M(V')$, which is easy to find.
\end{enumerate}

\section{Conclusion and open problems}
\label{sec:open}

We have shown the upper bound $\minvdeg(F) \le \nonmer(\surp(F))$ for lean clause-sets (Theorem \ref{thm:leanminvardeg}). The function $\nonmer(k)$ has been characterised in Lemma \ref{lem:solveN} and Corollary \ref{cor:upperboundnonmer}. We presented first initial results regarding the sharpness of the bound and regarding the constructive aspects of the bound (i.e., what happens if the bound is violated). There remain several open problems:
\begin{enumerate}
\item Prove Conjecture \ref{con:findauthard}, which says that such an autarky, which must exist if a clause-set does not fulfil the upper bound on the minimum variable degree of Theorem \ref{thm:leanminvardeg}, can be found in polynomial time. See Subsection \ref{sec:findaut} for more information on this topic.
\item Generalise Theorem \ref{thm:MUminvdegdef} to clause-sets with non-bool\-ean variables; see the discussion after Theorem \ref{thm:MUminvdegdef}.
\item See the remarks to Corollary \ref{cor:charsurp1} (an underlying question is to understand better the quantity ``surplus'').
\item Strengthen the bound on the minimum variable degree for minimally unsatisfiable clause-sets (see the forthcoming \cite{KullmannZhao2010Extremal}).
\item Strengthen the construction of Lemma \ref{lem:exmpmlean} (perhaps completely different constructions are needed).
\end{enumerate}

As mentioned in the introduction, a major motivation for us is the project of the classification of minimally unsatisfiable clause-sets for deficiencies $\delta = 1, 2, \dots$. Especially the classification of unsatisfiable hitting clause-sets in dependency on the deficiency seems very interesting (recall that a hitting clause-set $F$ is defined by the condition that every two clauses $C, C' \in F$, $C \not= C'$, clash in at least one variable, that is $\abs{C \cap \ol{C'}} \ge 1$). The main conjecture is:
\begin{conj}\label{con:finhit}
  For every deficiency $k \in \NN$ there are only finitely many isomorphism types of non-singular unsatisfiable hitting clause-sets.
\end{conj}
For $k \le 2$ this conjecture follows from known results, while recently we were able to prove it for $k=3$.

\bibliographystyle{plain}

\end{document}